\documentclass[a4paper]{article}
\usepackage{amsmath,amsfonts,amsthm, amssymb, amsbsy, url}

\newtheorem{theorem}{Theorem}
\newtheorem{lemma}[theorem]{Lemma}
\newtheorem{cor}[theorem]{Corollary}

\newcommand{\steg}{\operatorname{steg}}

\begin{document}

\title{Bounds on Herman's algorithm}
\author{John Haslegrave}

\maketitle

\begin{abstract}
Herman's self-stabilisation algorithm allows a ring of $N$ processors having any odd number of tokens to reach a stable state where exactly one token remains. McIver and Morgan conjecture that the expected time taken for stabilisation is maximised when there are three equally-spaced tokens. We prove exact results on a related cost function, and obtain a bound on expected time which is very close to the conjectured bound.
\end{abstract}


\section{Introduction}
Self-stabilisation algorithms were first discussed by Dijkstra \cite{Dij74}, and have since been widely studied (see eg \cite{Dol00}, \cite{Sch93}). Herman's algorithm provides a randomised self-stabilisation mechanism for $N$ processors connected unidirectionally in a ring, with synchronous updates. Each processor either has a token or does not. In an initial state an unknown, but odd, number of processors hold tokens, and the system is stable if there is only one token. Herman proposed a scheme where, simultaneously at each time step, each processor which has a token decides independently at random between keeping the token or passing it to the next processor in the ring, choosing each with equal probability. All updates occur simultaneously; if two tokens collide (because a processor keeping its token receives another from the previous processor) then they annihilate each other. Clearly the number of tokens will remain odd if this procedure is followed, and will never increase. It is easy to see that from any non-stable state with $N$ processors there is a probability of at least $2^{-2N/3}$ that an annihilation will occur within $\frac{1}{3}N$ steps. Consequently the algorithm stabilises almost surely; in fact the total time taken has finite expectation. The algorithm was introduced in \cite{Her90}, where Herman showed that the expected time to stabilisation is $O(N^2\log N)$. This bound was improved to $O(N^2)$ independently, with different constants, by McIver and Morgan \cite{MM05}; Fribourg, Messika, and Picaronny \cite{FMP05}; and Nakata \cite{Nak05}. McIver and Morgan also conjectured that the expected time is maximised by a starting state of three equally-spaced tokens; they show that this state has expected time $\frac{4}{27}N^2$ and that any other three-token configuration has a lower expected time. Kiefer, Murawski, Ouaknine, Wachter and Worrell \cite{KMOWW} extend this result by showing that the probability of stabilisation by time $t$ of any three-token state is bounded by that of the equally-spaced three-token state for each $t$. The conjecture of McIver and Morgan is supported by simulations using the PRISM model-checking software \cite{PRISM}. When there may be any number of initial tokens, the best previous upper bound is about $0.64N^2$, by Kiefer, Murawski, Ouaknine, Worrell and Zhang \cite{KMOWZ}. Here we give an upper bound of about $0.156N^2$, which is comparatively close to the conjectured value of just over $0.148N^2$. Our approach will differ from that of other papers in that we will in fact prove exact results for the expectation of a different cost function. A bound of $\frac{1}{6}N^2$ will immediately follow from the fact that our cost function is at least the time taken, and this can be slightly improved by considering the relationship between time and cost more carefully.

\section{The $\steg$ function and inequality}
For any odd $m\geqslant 3$ and variables $a_1,\ldots a_m$, a \textit{triple with even gaps} is a term which is the product of $a_i$, $a_j$ and $a_k$ for some $i<j<k$ and $k-j,k-i$ both odd (so that the number of unused variables between consecutive used variables is even). If we consider the variables as indexed by elements of $\mathbb{Z}_m$ then $a_{i+1}a_{j+1}a_{k+1}$ is a triple with even gaps if and only if $a_ia_ja_k$ is.

\begin{lemma}\label{number}There are $\frac{1}{24}m(m+1)(m-1)$ triples with even gaps on $a_1,\ldots, a_m$.\end{lemma}
\begin{proof}
The triples with even gaps containing $a_1$ are the triples of the form $a_1a_ja_k$ with $j<k$, $j$ even, and $k$ odd. There is a one-to-one correspondence between these terms and unordered pairs from the set $\{1,\ldots,\frac{m+1}{2}\}$, in which $a_1a_ja_k$ corresponds to the pair $\frac{j}{2},\frac{k+1}{2}$. Consequently there are $\binom{(m+1)/2}{2}$ such triples, and, since each variable is in the same number of triples with even gaps, there are $\frac{m}{3}\binom{(m+1)/2}{2}=\frac{1}{24}m(m+1)(m-1)$ triples with even gaps.
\end{proof}

Write $\steg(a_1,\ldots a_m)$ for the sum of all triples with even gaps (note that when $m=1$ there are no triples, and so $\steg(a_1)\equiv 0$). As we noted  above, cyclic shifts of the variables do not change which terms are included and so preserve $\steg$. The motivation for introducing this function is the following reduction when one of the variables is set to 0; we choose the penultimate variable for notational convenience.

\begin{lemma}\label{reduction}
For $m\geqslant 3$, if $a_{m-1}=0$ then
\[
\steg(a_1,\ldots a_m)=\steg(a_1,\ldots a_{m-3}, a_{m-2}+a_m)\,.
\]
\end{lemma}
\begin{proof}
In $\steg(a_1,\ldots a_m)$, all terms with $a_{m-1}$ vanish. No terms include both $a_{m-2}$ and $a_m$ apart from $a_{m-2}a_{m-1}a_m$, which vanishes. For $i<j<m-2$, $a_ia_ja_{m-2}$ is a triple with even gaps if and only if $a_ia_ja_m$ is. 

For $i<j<k<m-2$, $a_ia_ja_k$ is a term in the LHS if and only if it is a term in the RHS. The remaining terms on the LHS occur in pairs, with $a_ia_ja_{m-2}$ being such a term if and only if $a_ia_ja_m$ is, and that pair of terms appears in the LHS if and only if $a_ia_j(a_{m-2}+a_m)$ is a term in the RHS, so the two sides are equal.
\end{proof}
We next give an upper bound on $\steg$.

\begin{theorem}\label{inequality}For any odd $m\geqslant 3$, if $x_1,x_2,\ldots,x_m$ are non-negative reals such that $x_1+x_2+\cdots+x_m=1$ then 
\begin{equation}\label{ineq}
\steg(x_1,\ldots x_m)\leqslant\frac{1}{24}\left(1-\frac{1}{m^2}\right)\,.
\end{equation}
\end{theorem}
\begin{proof}
\[
\frac{1}{24}\left(1-\frac{1}{m^2}\right)=\frac{m(m+1)(m-1)}{24m^3}\,,
\]
and, by Lemma \ref{number}, $\steg(x_1,\ldots x_m)$ is the sum of $m(m+1)(m-1)/24$ terms. Each term takes the value $m^{-3}$ when $x_1=x_2\cdots =x_m=\frac{1}{m}$, so the RHS of \eqref{ineq} is just the value taken at this point. We will prove that the maximum cannot be attained anywhere else; since we are optimising a continuous function over a compact set, this will be sufficient. 

We use induction on $m$; the result holds for $m=3$ by the AM--GM inequality. If $m>3$ and $x_i=0$ for any $i$ then $\operatorname{steg}(x_1,\ldots x_m)=\operatorname{steg}(y_1,\ldots y_{m-2})$ for some non-negative $y_i$ summing to 1, by Lemma \ref{reduction}. By the induction hypothesis,
\begin{align*}
\operatorname{steg}(y_1,\ldots y_{l-2})&\leqslant\frac{1}{24}\left(1-\frac{1}{(m-2)^2}\right) \\
&<\frac{1}{24}\left(1-\frac{1}{m^2}\right) \\
&=\operatorname{steg}\left(\frac{1}{m},\ldots,\frac{1}{m}\right) \,,
\end{align*}
and so the maximum is not attained at any such point. 

For any point with $x_i>0$ for all $i$, consider the difference 
\[
\steg(x_1+\delta, x_2, x_3+\delta, x_4, x_5, \ldots x_m)-\steg(x_1,\ldots x_m)\,.
\]
Certainly any term which does not contain $x_1$ or $x_3$ will contribute nothing to the difference. For $j,k>2$, $x_1x_jx_k$ is a triple with even gaps if and only if $x_3x_jx_k$ is, and in that case $x_1x_jx_k+x_3x_jx_k=(x_1+\delta)x_jx_k+(x_3-\delta)x_jx_k$. So the only contributions to the difference are from terms of the form $x_1x_2x_3$, $x_1x_2x_k$ (with $k$ odd and greater than 3) and $x_2x_3x_k$ (with $k$ even). 
\begin{align*}
(x_1+\delta)x_2(x_3-\delta)-x_1x_2x_3&=\delta x_2x_3-\delta x_1x_2-\delta^2x_2\,; \\
(x_1+\delta)x_2x_k-x_1x_2x_k&=\delta x_2x_k\,; \\
\text{and} \quad x_2(x_3-\delta)x_k-x_2x_3x_k&=-\delta x_2x_k\,;
\end{align*}
so modifying $x_1$ and $x_3$ in this way increases the value of the function by $\delta x_2[(x_3+x_5+x_7+\cdots)-(x_1+x_4+x_6+\cdots)]-\delta^2x_2$. Since $x_1, x_2, x_3>0$, if $(x_3+x_5+x_7+\cdots)\neq(x_1+x_4+x_6+\cdots)$ then we can choose $\delta$ such that this is positive (and $x_1+\delta$ and $x_3-\delta$ are non-negative). Consequently if the maximum is attained at $x_1,\ldots x_l$ then we must have $(x_3+x_5+x_7+\cdots)=(x_1+x_4+x_6+\cdots)$. 
Applying the same argument to $x_2$ and $x_4$ shows that additionally we must have $(x_2+x_5+x_7+\cdots)=(x_1+x_4+x_6+\cdots)$, so $x_2=x_3$. Since the function is unchanged by a cyclic shift of the variables, we must have $x_i=x_{i+1}$ for every $i\in \mathbb{Z}_m$, so the maximum can only be attained when $x_i=\frac{1}{m}$ for every $i$.
\end{proof}
Note that many ostensibly similar functions do not satisfy an analogous inequality. If $f(x_1,\ldots x_m)$ is a sum of fewer than $m^3/27$ triples then it is not maximised when all variables are equal, yet $\steg$ has only slightly more than this. Also, if $g(x_1,\ldots x_m)$ is a sum of fewer than $3m^3/64$ triples which is maximised when all variables are equal then it cannot include three triples from any set of four variables.

Next we consider the average effect of a certain random transformation of the variables on $\steg(a_1,\ldots,a_m)$. Fix a subset $S\subset [m]$ of even size $2h$, and write $S=\{i_1,\ldots i_{2h}\}$ with $i_1<i_2<\cdots<i_{2h}$. The move $M_S^+$ adds 1 to $a_{i_k}$ for each odd $k$ and subtracts 1 from $a_{i_k}$ for each even $k$. The move $M_S^-$ does the opposite. Note that $M_{\varnothing}^+$ and $M_{\varnothing}^-$ both leave $a_1\ldots a_m$ unchanged; nevertheless, we regard them as different moves. There are then $2^m$ possible moves: 2 moves for each of the $2^{m-1}$ even subsets. 
\begin{theorem}\label{moves}
Writing $\boldsymbol{a}=(a_1,\ldots,a_m)$, let $\tilde{\boldsymbol{a}}$ be the random vector obtained by applying one of the $2^m$ possible moves, chosen uniformly at random, to $\boldsymbol{a}$. Then
\[
\mathbb{E}(\steg(\tilde{\boldsymbol{a}}))=\steg(\boldsymbol{a})-\frac{m-1}{8}\sum\nolimits_ka_k\,.
\] 
\end{theorem}
\begin{proof}
Note that 
\[
2^m\big(\steg(\boldsymbol{a})-\mathbb{E}(\steg(\tilde{\boldsymbol{a}}))\big)
=\sum_S\big(2\steg(\boldsymbol{a})-\steg\left(M_S^+\boldsymbol{a}\right)-\steg\left(M_S^-\boldsymbol{a}\right)\big)\,,
\]
so it will be sufficient to show that the latter is equal to $2^{m-3}(m-1)\sum_ka_k$.

Fix a triple with even gaps $a_ia_ja_k$, ordering the variables so as to preserve the cyclic ordering within $Z_m$ (that is to say, regarded as members of the set $\{1,\ldots,m\}$, either $i<j<k$, $j<k<i$, or $k<i<j$). Consider the contribution of this triple to $2\steg(\boldsymbol{a})-\steg\left(M_S^+\boldsymbol{a}\right)-\steg\left(M_S^-\boldsymbol{a}\right)$; call this quantity $\beta_{ijk}^S$. The move $M_S^+$ adds $\eta_i$ to $a_i$, for some $\eta_i \in \{-1,0,+1\}$, and $M_S^-$ subtracts the same amount. Note that $\eta_i$ is non-zero if and only if $a_i\in S$. With the same notation for $j$ and $k$,
\begin{align*}
\beta_{ijk}^S&=2a_ia_ja_k-(a_i+\eta_i)(a_j+\eta_j)(a_k+\eta_k)-(a_i-\eta_i)(a_j-\eta_j)(a_k-\eta_k) \\
&=-2(\eta_i\eta_ja_k+\eta_j\eta_ka_i+\eta_k\eta_ia_j) \,.
\end{align*}
We will calculate the sum $\sum_S\beta_{ijk}^S$; note that
\[
\sum_S\big(2\steg(\boldsymbol{a})-\steg\left(M_S^+\boldsymbol{a}\right)-\steg\left(M_S^-\boldsymbol{a}\right)\big)=\sum_{\text{t.e.g.}}\left(\sum_S\beta_{ijk}^S\right)\,,
\]
where the outer sum is taken over all triples with even gaps. 

Write $X_{ij}=\{i,\ldots,j\}\setminus\{i,j\}$ (and similarly define $X_{jk}$, $X_{ki}$); recall that we chose our ordering of the variables such that $k\notin X_{ij}$. We distinguish three cases according to the number of $i,j,k$ that $S$ contains.

\noindent\textbf{Case 1} If $S$ contains at most one of $i,j,k$ then $\beta_{ijk}^S=0$.
 
\noindent\textbf{Case 2} If $S$ contains two of $i,j,k$, without loss of generality $i,j\in S$ but $k\notin S$, then $\beta_{ijk}^S=-2a_k$ if $\eta_i$ and $\eta_j$ have the same sign, and $\beta_{ijk}^S=2a_k$ if they have opposite signs. Write $S_1=S\cap X_{ij}$ and $S_2=S\cap (X_{jk}\cup X_{ki})$, so that $S={i,j}\cup S_1\cup S_2$. Now $|S_1|\equiv|S_2|$ mod 2; $\eta_i$ and $\eta_j$ have the same sign if and only if these cardinalities are both odd. A non-empty finite set has exactly half its subsets odd, so if both $X_{ij}$ and $X_{jk}\cup X_{ki}$ are non-empty there are $2^{m-5}$ choices of $S$ for which $\beta_{ijk}^S=-2a_k$ and $2^{m-5}$ for which $\beta_{ijk}^S=2a_k$. If one of the sets is empty then either $i,j$ are consecutive mod $m$ or $j,k,i$ are consecutive mod $m$; in either of these cases each of the $2^{m-4}$ possible choices of $S$ has $\beta_{ijk}^S=2a_k$.

\noindent\textbf{Case 3} If $i,j,k\in S$ then write $S_1=S\cap X_{ij}$, $S_2=S\cap X_{jk}$, $S_3=S\cap X_{ki}$. $|S_1|+|S_2|+|S_3|$ must be odd, so the possibilities are:
\begin{enumerate}
\item all three are odd, when $\beta_{ijk}^S=2(-a_i-a_j-a_k)$;
\item only $|S_1|$ is odd, when $\beta_{ijk}^S=2(a_i+a_j-a_k)$;
\item only $|S_2|$ is odd, when $\beta_{ijk}^S=2(-a_i+a_j+a_k)$; and
\item only $|S_3|$ is odd, when $\beta_{ijk}^S=2(a_i-a_j+a_k)$.
\end{enumerate}
If $X_{ij}$, $X_{jk}$ and $X_{ki}$ are all non-empty then each possibility occurs in $2^{m-6}$ ways. If only $X_{ij}$ is empty then (i) and (ii) are impossible; there are then exactly $2^{m-5}$ ways for each of (iii) and (iv) to occur. If $X_{ij}$ and $X_{jk}$ are empty then only (iv) is possible and it occurs in $2^{m-4}$ ways.

\noindent\textbf{End of cases}

Overall, then, a term which has no two variables consecutive contributes nothing. 
A term where exactly two variables, without loss of generality $a_i$ and $a_j$, are consecutive, contributes $2^{m-3}a_k$ from choices of $S$ which include $i$ and $j$ but not $k$, and another $2^{m-3}a_k$ from choices of $S$ which include all three, for a total of $2^{m-2}a_k$. For each $k$, there are $\frac{m-5}{2}$ triples with even gaps which have this property, so terms of this form contribute $(m-5)2^{m-3}\sum_ka_k$ in total.
 
A term where all three variables are consecutive, i.e. one of the form $a_{j-1}a_ja_{j+1}$, contributes $2^{m-3}(a_{j-1}-a_j+a_{j+1})$ from choices of $S$ which include $j-1$, $j$ and $j+1$; $2^{m-3}a_{j-1}$ from choices of $S$ which include $j$ and $j+1$ but not $j-1$; $2^{m-3}a_{j+1}$ from choices of $S$ which include $j-1$ and $j$ but not $j+1$; and $2^{m-3}a_{j}$ from choices of $S$ which include $j-1$ and $j+1$ but not $j$. In total, then, such a term contributes $2^{m-3}(2a_{j-1}+2a_{j+1})$; since there is one term of this form for each $j$, they contribute $4\times 2^{m-3}\sum_ka_k$ in total. Thus
\[
\sum_{\text{t.e.g.}}\left(\sum_S\beta_{ijk}^S\right)=(m-1)2^{m-3}\sum_ka_k\,,
\]
as required.
\end{proof}

\section{Relating Herman's algorithm to $\steg$}
Run Herman's algorithm with $N$ processors from some starting state $A^{(0)}$ (with an odd number of tokens), to get a sequence of states $(A^{(t)})_{t\geqslant 0}$. For each time step we incur a cost: if there are $m$ tokens at time $t$ the step from $t$ to $t+1$ has cost $\frac{m-1}{2}$. Also write $c_t$ for the cost accumulated by time $t$, so that $c_0=0$ and if $A^{(t)}$ has $m$ tokens then $c_{t+1}=c_t+\frac{m-1}{2}$. Note that $c_t$ is (with probability 1) ultimately constant, since it stops increasing once a stable state is reached. We shall use the results proved in Section 2 to give an exact value for the expected total cost.

With each state we associate a vector $\boldsymbol{a}^{(t)}$, whose components are the distances between consecutive tokens. We define this more precisely as follows: if $A^{(t)}$ has $m$ tokens then write $b^{(t)}_1<\cdots<b^{(t)}_m$ for their positions, then $\boldsymbol{a}^{(t)}$ is the vector of $m$ components with 
\[
a^{(t)}_i=
\begin{cases}b^{(t)}_{i+1}-b^{(t)}_i &\mbox{for } i<m \\
N+b^{(t)}_1-b^{(t)}_i &\mbox{for } i=m\;.
\end{cases}
\]
Note that $\sum_ia^{(t)}_k=N$, the total number of processors. Now define a sequence of variables $(X_t)_{t\geqslant 0}$ as $X_t=\steg(\boldsymbol{a}^{(t)})+\frac{1}{4}Nc_t$. 

\begin{theorem}The sequence $(X_t)_{t\geqslant 0}$ is a martingale, in the sense that
\[
\mathbb{E}(X_{t+1}\mid A^{(0)},\ldots,A^{(t)})=X_t\,.
\]
\end{theorem}
\begin{proof}
Herman's algorithm is a Markov chain so
\[
\mathbb{E}(X_{t+1}\mid A^{(0)},\ldots,A^{(t)})=\mathbb{E}(X_{t+1}\mid A^{(t)})\,.
\]
Fix $A^{(t)}$ and suppose that it has $m$ tokens at positions $b^{(t)}_1<\cdots<b^{(t)}_m$. Write $\hat{b}^{(t+1)}_1,\ldots,\hat{b}^{(t+1)}_m$ for the positions of the corresponding tokens at time $t+1$ before any annihilations occur, and define gaps $\hat{a}^{(t+1)}_1,\ldots,\hat{a}^{(t+1)}_m$ as above.
While $\boldsymbol{a}^{(t+1)}$ is not in general equal to $\hat{\boldsymbol{a}}^{(t+1)}$, not only because some of the terms may be 0 but also because the latter will be cyclically shifted if a token has moved from position $N$ to position $1$, we claim that $\steg(\boldsymbol{a}^{(t+1)})=\steg(\hat{\boldsymbol{a}}^{(t+1)})$. 
Cyclic shifts do not change $\steg$, so only the collisions need concern us. Note that it is not possible for $\hat{a}^{(t+1)}_i$ and $\hat{a}^{(t+1)}_{i+1}$ to both be 0, as the former is only possible if the token at $b^{(t)}_i+1$ did not move and the latter if it did. 
If there is a single collision, say $\hat{a}^{(t+1)}_i=0$, then $a^{(t+1)}_j=\hat{a}^{(t+1)}_j$ for $j<i-1$, $a^{(t+1)}_{i-1}=\hat{a}^{(t+1)}_{i-1}+\hat{a}^{(t+1)}_{i+1}$, and $a^{(t+1)}_j=\hat{a}^{(t+1)}_{j+2}$ for $j>i+1$. By Lemma \ref{reduction}, $\steg(\boldsymbol{a}^{(t+1)})=\steg(\hat{\boldsymbol{a}}^{(t+1)})$. If there are multiple collisions we may carry out each one in turn, and $\steg$ will be preserved at each step.

Write $S^+$ for the set of $i$ such that the token at $b^{(t)}_i$ does not move but the token at $b^{(t)}_{i+1}$ does, and $S^-$ for the set of $i$ such that the token at $b^{(t)}_i$ moves but the token at $b^{(t)}_{i+1}$ does not. Writing $S=S^+\cup S^-$, the elements of $S$ must alternate between the two types and so $|S^+|=|S^-|$. Now
\[
\hat{a}^{(t+1)}_i=
\begin{cases}a^{(t)}_i+1 &\mbox{if } i\in S^+ \\
a^{(t)}_i-1 &\mbox{if } i\in S^- \\
a^{(t)}_i &\mbox{otherwise,}
\end{cases}
\]  
so, if $S$ is non-empty, $\hat{\boldsymbol{a}}^{(t+1)}=M_S^+(\boldsymbol{a}^{(t)})$ if the smallest member of $S$ is in $S^+$, and $\hat{\boldsymbol{a}}^{(t+1)}=M_S^-(\boldsymbol{a}^{(t)})$ otherwise. For each non-empty even set $S$ there are two possibilities for $(S^+,S^-)$, and each of these uniquely determines which of the tokens moves, so has probability $2^{-m}$. If $S=\varnothing$ then either all tokens move or no tokens move; we may think of these as corresponding to $M_\varnothing^+$ and $M_\varnothing^-$ respectively. So $\hat{\boldsymbol{a}}^{(t+1)}$ is obtained from $\boldsymbol{a}^{(t)}$ by applying one of the $2^m$ moves, and each is equally likely. Applying Theorem \ref{moves}, 
\begin{align*}
\mathbb{E}(X_{t+1}\mid A^{(t)})&=\mathbb{E}(\steg(\boldsymbol{a}^{(t+1)})\mid A^{(t)})+\frac{N}{4}\mathbb{E}(c_{t+1}\mid A^{(t)}) \\
&=\mathbb{E}(\steg(\hat{\boldsymbol{a}}^{(t+1)})\mid \boldsymbol{a}^{(t)})+\frac{N}{4}\left(\frac{m-1}{2}+c_t\right) \\
&=\steg(\boldsymbol{a}^{(t)})-\frac{m-1}{8}N+\frac{N}{4}\left(\frac{m-1}{2}+c_t\right) \\
&=X_t\,,
\end{align*}
as required.
\end{proof}
We now use the fact that $(X_t)$ is a martingale to deduce the exact value of the expected total cost.

\begin{theorem}
The expected total cost starting from the state $A^{(0)}$ is $\frac{4}{N}\steg(\boldsymbol{a}^{(0)})$.
\end{theorem}
\begin{proof}
Let $T$ be the earliest time for which $A^{(T)}$ is stable. $T$ is a stopping time, and, as we observed in Section 1, it has finite expectation. Further, if $a_1,\ldots,a_m$ are non-negative with sum $N$ then 
\begin{align*}
\steg(a_1,\ldots,a_m)&=N^3\steg(a_1/N,\ldots,a_m/N) \\
&\leqslant N^3\frac{1}{24}\left(1-\frac{1}{m^2}\right) \\
&<\frac{N^3}{24} \,,
\end{align*}
by Theorem \ref{inequality}, and $\steg(a_1,\ldots,a_m)\geqslant 0$, so
\begin{align*}
|X_{t+1}-X_t|&\leqslant|\steg(\boldsymbol{a}^{(t+1)})-\steg(\boldsymbol{a}^{(t)})|+\tfrac{1}{4}N(c_{t+1}-c_t) \\
&<\frac{N^3}{24}+\frac{N^2}{8} \,.
\end{align*}
Since the stopping time has finite expectation and there is a global bound on the difference between successive variables, the Optional Stopping Theorem (see, for example, \cite{Wil}, p. 100) applies and so $\mathbb{E}(X_T)=X_0$. Consequently
\[
\mathbb{E}(\steg(\boldsymbol{a}^{(T)}))+\frac{N}{4}\mathbb{E}(c_T)=\steg(\boldsymbol{a}^{(0)})\,.
\]
Since $A^{(T)}$ is stable, $\boldsymbol{a}^{(T)}$ has only one component, and so $\steg(\boldsymbol{a}^{(T)})=0$ (with probability 1); also $c_T$ is the total cost since no further cost is incurred after time $T$. So we have
\[
\mathbb{E}(c_T)=\frac{4}{N}\steg(\boldsymbol{a}^{(0)})\,,
\]
as required.
\end{proof}
Since $\steg(\boldsymbol{a}^{(0)})=N^3\steg(a^{(0)}_1/N,\ldots,a^{(0)}_m/N)$, applying Theorem \ref{inequality} gives the following result.
\begin{cor}\label{bound}The expected total cost from any starting state with $2s+1$ tokens is at most $\left(1-(2s+1)^{-2}\right)N^2/6$.
\end{cor}
We are now ready to prove our main result.
\begin{theorem}The expected time to stabilisation for Herman's algorithm from any start state on $N$ processors is less than $(\pi^2-8)N^2/12$.\end{theorem}
\begin{proof}
Fix a start state $A$ with $2r+1$ tokens, and let the random variable $T$ be the total time to stabilisation. For each $s\geqslant 1$, write $A_s$ for the first configuration with at most $2s+1$ tokens and $C_s$ for the cost accumulated after that point. Note that $A_s=A_r=A$ and $C_s=C_r$ for every $s\geqslant r$. Now $C_0=0$ and $C_s-C_{s-1}=0$ for all $s\geqslant r$, so for any $t\geqslant r$,
\begin{align*}
T&=\sum_{s=1}^r\tfrac{1}{s}\left(C_s-C_{s-1}\right) \\
&=\sum_{s=1}^t\tfrac{1}{s}\left(C_s-C_{s-1}\right) \\
&=\sum_{s=1}^t\left(\tfrac{1}{s}-\tfrac{1}{s+1}\right)\!C_s+\tfrac{1}{t+1}C_t \\
&=\sum_{s=1}^t\tfrac{1}{s(s+1)}C_s+\tfrac{1}{t+1}C_t,.
\end{align*}
Since
\[
\lim_{t\to\infty}\tfrac{1}{t+1}C_t=0\,,
\]
it follows that
\[
\sum_{s=1}^\infty\tfrac{1}{s(s+1)}C_s=T\,.
\]
Also,
\begin{align*}
\mathbb{E}(C_s)&=\mathbb{E}(\mathbb{E}(C_s\mid A_s)) \\
&\leqslant \left(1-\frac{1}{(2s+1)^2}\right)\!\frac{N^2}{6}\,,
\end{align*}
by Corollary \ref{bound}, and so
\begin{align*}
\mathbb{E}(T)&\leqslant\sum_{s=1}^\infty\frac{1}{s(s+1)}\left(1-\frac{1}{(2s+1)^2}\right)\!\frac{N^2}{6} \\
&=\sum_{s=1}^\infty\frac{1}{s(s+1)}\left(\frac{4s^2+4s}{(2s+1)^2}\right)\!\frac{N^2}{6} \\
&=\raisebox{-0.5ex}{$\dfrac{2N^2}{3}$}\sum_{s=1}^\infty\frac{1}{(2s+1)^2} \\
&=\frac{(\pi^2-8)}{12}N^2\,,
\end{align*}
as required.
\end{proof}

\section{Acknowledgements}
The author acknowledges support from the European Union
through funding under FP7--ICT--2011--8 project HIERATIC (316705).

\end{document}